\documentclass[A4paper,11pt]{article}
\usepackage{amsmath,amssymb,amsfonts}

\newtheorem{theorem}{Theorem}[section]
\newtheorem{proposition}[theorem]{Proposition}
\newtheorem{corollary}[theorem]{Corollary}

\newtheorem{definition}[theorem]{Definition}

\newtheorem{example}[theorem]{Example}

\newenvironment{proof}{\mbox{\bf Proof.}}{\mbox{$\dashv$}\bigskip}

\begin{document}
    
    \begin{center}

        {\Large\bf   Computer-Simulation Model Theory}  \\  {(    $  \mathrm{ P= NP}$ is not provable)    }\\~\\

        {\bf  Rasoul Ramezanian}\\
        {\bf  Mashhad, Iran}~\footnote{rramezanian@um.ac.ir, ramezanian@sharif.ir}\\

    \end{center}

\begin{abstract} 

\noindent The simulation hypothesis says that all the materials and events in the reality (including the universe, our body, our thinking, walking and etc) are computations, and the reality is a computer simulation program like a video game. All works we do (talking, reasoning, seeing and etc) are computations performed by the universe-computer which runs the simulation program.
 
  \noindent Inspired by the view of the simulation hypothesis  (but independent of this hypothesis), we propose a new method of logical reasoning named "Computer-Simulation Model Theory", CSMT. 
 
\noindent Computer-Simulation Model Theory is an extension of Mathematical Model Theory where instead of mathematical-structures, computer-simulations are replaced, and the activity of reasoning and computing of the reasoner is also simulated in the model.   (CSMT) argues that:

\begin{quote}
  \textit{for a formula $\phi$, construct  a computer simulation model $S$} such that
\begin{itemize}
        \item[1-]    \textit{$\phi$ does not hold in $S$, and}
    \item[2-] \textit{the reasoner $I$ $($human being, the one who lives inside the reality$)$ cannot distinguish $S$ from the reality $(R)$,  }

\end{itemize}  \textit{then}   
  
\textit{$I$ cannot prove $\phi$ in reality.}
\end{quote}
 \noindent Although  $\mathrm{CSMT}$  is inspired by the simulation hypothesis, but this reasoning method is independent of the acceptance of this hypothesis. 
 As we argue in this part, one may do not accept the simulation hypothesis, but knows $\mathrm{CSMT}$ a valid reasoning method.  
 
 As an application of Computer-Simulation Model Theory, we study the famous problem P vs NP. We let $\phi \equiv\mathrm{ [P= NP]} $ and construct  a computer simulation model $E$  such that     $\mathrm{ P= NP}$ does not hold in $E$.    
\end{abstract}

\section{Simulation Hypothesis}
\noindent 
Human's worldview changes with the times. At the time of Omar Khayyam, the industry was limited to the Pottery, and Omar Khayyam (inspired by the pottery) assumed the universe as a water vessel and asked "hey you! where is the potter?"

 Our worldview changes as our industry changes and our current industry are artificial intelligence, virtual reality, robotics, and video games. Using artificial intelligence, we are going to make new creatures even more intelligent than us. Sophia Robot~\footnote{$https://en.wikipedia.org/wiki/Sophia_(robot)$} is one of these creatures.

 So Nick Bostrom (inspired by  the computer industry) asks "Are you living in a computer simulation?" (see~\cite{kn:NBOS});   we might be all creatures in a video game developed by a programmer.
We almost are going to believe this new worldview as Elon Musk says that he feels "We're Probably Living in a Simulation".

 Briefly, the simulation hypothesis is the following sentence
 \begin{quote}  
     we are living in a simulation program (a video game) which we, inhabitants, call it Reality. This simulation program is run on a computer device (similar to your PC or your laptop)  which we call it universe-computer (UC).
\end{quote}

Regarding  reality as a computer simulation means that all we (human beings) do (thinking, walking, reasoning, talking, jumping, and etc) are computations and procedures performed by the "Universe-Computer" (UC). When I think, it is UC that computes, when I prove a theorem, it is UC that computes, when I breathe, it is again UC, and etc. Physics, Biology, Logic and etc are all computational process which UC performs.

So, we may name human beings (the one who lives in) to be the \underline{computist} as his/her thinking, breathing, walking and etc are all computations done by UC.

In the next section, we discuss our proposed reasoning method $\mathrm{CSMT}$. For justification of our reasoning method, we do not need the simulation hypothesis to be "True", we just need it to be "Possible".  So one may do not accept the simulation hypothesis, but accepts the $\mathrm{CSMT}$ as a valid reasoning method.

\section{Computer-Simulation Model Theory}

Computer-Simulation Model Theory is an extension of the Model Theory~\cite{kn:modelm}, where the reasoner also lives inside the model and is an inhabitant of the model.

 Model theory is a branch of mathematical logic that studies that mathematical-structures satisfy which formulas.  Given a set of formulas $\Gamma$, and a formula $\phi$, a way to reason that  $\phi$ is not logically derivable from $\Gamma$ is simply by constructing a mathematical-structure    $M$, such that all formulas in $\Gamma$ holds in $M$ ($M\models \Gamma$) but $\phi$ does not hold in the world $M$ ($M\not\models \phi$). This reasoning is true, because of the soundness of deduction system~\footnote{For example, suppose that we want to show that a formula $\phi$ is not True for all graphs. We only need to construct a graph structure which does not satisfy $\phi$.}.
 
 In Computer-Simulation Model Theory, instead of mathematical-structure, we deal with computer-simulations.   In model theory, the reasoner (the one who thinks and reasons) is out of the mathematical-structure, while   In CSMT, the reasoner (the one who thinks, computes, and ...) lives and simulated inside the computer-simulation. 
 
In a computer-simulation model, a formula $\psi$ could be provable for the creator of the computer-simulation but not provable for the computist. But we assume soundness, that is if the computist proves a formula $\phi$ then the formula is true in the computer-simulation.

The Computer-Simulation Model Theory  argues that:
 
 \begin{quote}
  to show that the reasoner (the computist) cannot prove a sentence $\phi$ in the reality ($R$),  construct a computer simulation $S$ such that
     \begin{itemize}
         \item[(1)] the reasoner  (computist) cannot distinguish $S$ from the reality $(R)$, and
         \item[(2)]    $\phi$ does not hold in $S$.
     \end{itemize} 
 \end{quote} 

If our deduction system is sound then all we prove must hold in reality. So, if $S$ is a computer simulation that we cannot say whether we live in $S$ or in the reality, then if we can prove  $\phi$ in the reality, the statement $\phi$ must also be true at $S$.

Note that our activity of proving and reasoning are computations performed by the UC of reality. This property differs CSMT from model theory.

  The activity of the reasoning and computing of the computist is also performed by the simulation, and this activity can affect the model.

\noindent We are in the age of artificial intelligence, and constructing a computer simulation exactly similar to the reality is completely plausible. Using artificial intelligence, we can simulate Physics, Biology, our thinking, our reasoning and etc.

 The mathematical Model theory considers the reasoner outside of the model. Model theory is not a complete and real reasoning method since in reality, we (the reasoner) are inhabitant and inside the model. In CSMT, an extension of the mathematical model theory,    the \textit{process of reasoning} is also considered in the model, and this process may affect other parts of the model.
  
  \subsection{Computer-Simulation Model}
   A computer-simulation model $S$ consists of the three following parts
  \begin{itemize}
      \item[i.]  The UC of $S$.
      \item[ii.] The Computist. 
      \item[iii.] A set of instruction  $INST$ that the computist (the inhabitants in $S$) interacts with UC through them. [when we walk, reason, breathe, and etc, we are asking UC   to perform the computation of walking, reasoning, breathing, and etc].
  \end{itemize}
  
  All we do (walking, reasoning, and etc) are procedures performed by UC of the simulation. So, to formally describe a computer-simulation model,  we just need to define 
  \begin{itemize}
      \item what a procedure is?
      \item how the computist using the procedures interacts with the UC of the computer-simulation model?
      
  \end{itemize} In the  next definitions~\ref{ucuc},~\ref{proced}, and~\ref{CE}, we clarify these notions.
  \begin{definition} \label{ucuc}
      The \emph{UC}  of a computer-simulation model $S$ is a tuple    $$U=(TBOX,SBOX, INST,
      CONF)$$ 
      where
      \begin{itemize}
          \item[\emph{1.}] $INST$ is a nonempty set \emph{(the set of all
          instructions)}, and $INST_0\subseteq
          INST$ is a nonempty subset called the set of  starting
          instructions.
          
          \item[\emph{2.}] $CONF$ is a nonempty set called the set of
          \emph{configurations} such that to each $x\in \{0,1\}^*$,
          \begin{itemize}
              \item  a unique configuration $C_{0,x}\in CONF$ is associated as the start
              configuration, and
              \item to each $C\in CONF$,  a unique string $y_C\in
              \{0,1\}^*$ is associated.
          \end{itemize}

      \end{itemize}
  
      \begin{itemize}
          \item[\emph{3.}]  \emph{The transition   box}, $TBOX$, is a total function
          from $CONF\times INST$ to $CONF\cup\{\bot\footnote{undefined symbol.}\}$. The function $TBOX$ is executable by a computer device~\footnote{We assume that the computer device is equipped and extended with memory cards as much as needed and we never face the shortage of memory, and thus Turing machines are also executable by computer devices.}. 
          
          \item[\emph{4.}] \emph{The successful   box}, $SBOX$,  is a total function
          from $CONF$ to $\{YES, NO\}$. The function $SBOX$ is executable by a computer device. 
      \end{itemize}  
  \end{definition}
  Note that when we say a computer-simulation model, then the simulation program of the model must be executable by a computer device and because of this, two functions $SBOX$ and $TBOX$ must be executable by a computer device.
  \begin{definition}\label{proced}
      ~\begin{itemize}
          
          \item[i.] A   \emph{procedure (an algorithm)}  is  a finite  set
          $M\subseteq INST$ \emph{(a finite set of instructions)}, satisfying the
          following condition 
          
          \begin{itemize} \item[] \emph{The determination condition}: for
          every $C\in CONF$ either for all $\iota \in M$,
          $TBOX(C,\iota)=\bot$, or at most there exists one instruction
          $\tau \in M$ such that $TBOX(C,\iota)\in CONF$. 
      \end{itemize}
          We refer to the
          set of all procedures by the symbol $\Xi$.
          
          \item[ii.]We let $\upsilon: \Xi\times CONF\rightarrow
          INST\cup\{\bot\}$ be a total function such that for each
          procedure $M$ and $C\in CONF$,
          if $\upsilon(M,C)\in INST$
          then 
          \begin{itemize}
              \item $\upsilon(M,C)\in M$, and
              \item $TBOX(C,\upsilon(M,C))\in CONF$.
          \end{itemize}
          The function $\upsilon$ controls that which instruction of a procedure $M$ must apply on a given configuration $C$. 
      \end{itemize} 
  \end{definition}
  
  \begin{definition} \label{CE} For every computer-simulation model $S$, we consider the followings to be true:
      
      \begin{itemize} 
          
          \item[$c1.$] \emph{Free will of Inhabitants:} The computist is free to do the
          following things in  any order that he wants:
          
          \begin{itemize} \item[$1-$] he can freely choose an arbitrary instruction $\iota \in INST$
              and an arbitrary configuration $C\in CONF$ to apply the $TBOX$ on
              $(C,\iota)$, and
              
              \item[$2-$] he  can freely choose an arbitrary  configuration
              $C\in CONF$ to apply the $SBOX$ on. 
          \end{itemize} 
          
          \item[$c2.$] \emph{Computable Languages:} A string $x\in \Sigma^*$, $\Sigma= \{0,1\}$,  is in the
          \emph{language} of a  procedure $M$, denoted by $L(M)$, whenever the computist can
          construct a sequence $C_{0}C_{1},...,C_{n}$ of configurations in
          $CONF$ such that
          \begin{itemize}\item $C_0=C_{0,x}$, \item each $C_i$, $i\geq 1$,
              is obtained by applying $TBOX$ on $(C_{i-1},\upsilon(M,C_{i-1}))$,
              \item the $SBOX$ outputs $YES$ for $C_n$, \item and either
              $\upsilon(M,C_{n})=\bot$ or $TBOX(C_n, \upsilon(M,C_{n}))=\bot$.
          \end{itemize}  The computist  calls $C_{0}C_{1},...,C_{n}$ the
          successful  computation path of $M$ on $x$. The length of a
          computation path is the number of configurations appeared in.

          \item[$c3.$] \emph{Computable Functions:} A partial function
          $f:\Sigma^*\rightarrow\Sigma^*$, $\Sigma=\{0,1\}$, is computed by a procedure $M\in
          \Xi$, whenever for $x\in \Sigma^*$,   the computist can construct a sequence
          $C_{0}C_{1},...,C_{n}$ of configurations in $CONF$ such that
          \begin{itemize}\item $C_0=C_{0,x}$, \item each $C_i$, $i\geq 1$,
              is obtained by applying $TBOX$ on $(C_{i-1},\upsilon(M,C_{i-1}))$,
              \item the $SBOX$ outputs $YES$ for $C_n$, \item and either
              $\upsilon(M,C_{n})=\bot$ or $TBOX(C_n, \upsilon(M,C_{n}))=\bot$,
              
              \item $y_{C_n}=f(x)$.
          \end{itemize}

          \item[$c4.$]\emph{Parallel use of UC:} The computist may start to apply the \emph{UC} on a   procedure  $M$ and a string $x$, however, he does not have to keep on the computation until the
          successful box outputs $Yes$. The computist can leave the
          computation $M$ on $x$ at any stage of his activity  and choose
          freely any other  procedure $M'$ and any other string $x'$
          to apply \emph{UC} on them.

          \item[$c5.$] \emph{Time Complexity:} The   \emph{time
              complexity} of computing a  procedure $M$ on an input string
          $x$, denoted by $time_M(x)$,   is $n$, for some $n\in
          \mathbb{N}$, whenever the computist constructs  a successful computation path
          of the syntax-procedure $M$ on $x$ with length $n$.

          \item[$c6.$] \emph{Time Complexity:} Let $f:\mathbb{N}\rightarrow \mathbb{N}$  and
          $L\subseteq \Sigma^*$. The computist says that the time complexity
          of the computation of the language $L$  is less than $f$ whenever there exists a
          procedure  $M\in \Xi$ such that the language defined by the
          computist via $M$, i.e., $L(M)$, is equal to $L$, and for all
          $x\in L$, $time_M(x)<f(|x|)$.
          
          \item[$c7.$] \emph{Complexity Classes:} The computist  defines the time complexity class
          $\mathrm{P}\subseteq 2^{\Sigma^*}$ to be the set of all
          languages that he/she can computes  in polynomial time. He/She
          also defines the complexity class $\mathrm{NP}\subseteq
          2^{\Sigma^*}$ as follows:
          \begin{itemize}
              \item[] $L\in \mathrm{NP}$  iff there exists $J\in \mathrm{P}$
              and  a polynomial function $q$ such that for all $x\in
              \Sigma^*$,\begin{center} $x\in L\Leftrightarrow\exists y\in
                  \Sigma^* (|y|\leq q(|x|) \wedge (x,y)\in J)$.\end{center}
          \end{itemize}
          \item[$c8.$] \emph{Turing Computability:} The \emph{UC} of the simulation is sufficiently  powerful such that the computist
          can compute    all partial  recursive (Turing computable) function using the \emph{UC}. That is, for every Turning machine $T$, there exists a procedure $M\in \Xi$ that is $L(M)=L(T)$. 
          
      \end{itemize}

      \begin{itemize}
          \item[$c9.$] \emph{Black Box:}  The computist lives inside the simulation $S$ and does not access to the structure of  $TBOX$ and $SBOX$. For the computist, $TBOX$ and $SBOX$ are black boxes.
          
          \item[$c10.$] \emph{ Deduction System:} There exist a procedure  $G\in \Xi$  such that  for every proof $z$, and every formula $\psi$, $G(\langle z,\psi' \rangle )= 1$ means $z$ is the proof of $\psi$. The computist using the procedure $G$ does his logical reasoning. Note that the procedure $G$ is also computed by the universe-computer UC.
          
          \item[$c11.$] \emph{Soundness:} If the computist can prove a sentence $\phi$, then $\phi$ is True in the simulation.
          
          \item[$c12.$] The universe-computer of the simulation, $UC$, works in linear time. Also, the computist knows that the UC works in linear time. There exists a universal clock in the simulation $S$, shown by $Clock_S$, such that the computist uses it and measures passing time.  That is when the computist insert a configuration $C$ to $SBOX$ ($TBOX$)  the number of clocks of $Clock_S$ that the computist waits to receive the output is linear with respect to the length of $C$.  
          
      \end{itemize}
  \end{definition}

 In c10, we insist on proof procedure $G$, we could similarly talk about "walking procedure", "breathing procedure" and etc, but we disregard them as we aim to construct a counter-model for $\mathrm{ [P= NP]}$. As an example of a computer-simulation model, one may see the example~\ref{tce}. 
 
\subsection{Indistinguishability}
 In the following definitions, we formally describe "indistinguishability" between computer-simulation models.
 
 \begin{definition}
     Suppose $S$ is a simulation and $M$ is a procedure in $S$. The experience of the computist on the procedure $M$, denoted by $EXP_S(M)$, is defined to be the set of all pairs $(x,y)$ where the computist ran procedure $M$ on input $x$ and received output $y$. Note that always, $EXP_S(M)$ is a finite set (although new pairs are always added to it) since, at each stage of time, the computist could only run the procedure on only finite inputs.
  
 \end{definition}

\begin{definition}
 Suppose $S$ and $S'$ are two computer-simulation model where UC is the universe-computer of $S$ and $UC'$ is the universe-computer of $S'$.
 We say that the computist $A$ who is an inhabitant in $S$ cannot say if he lives inside $S$ or $S'$, $S\sim_A S'$, whenever
 \begin{itemize}
     \item[1-] $INST_S$ = $INST_{S'}$,
     \item[2-] At each stage of time, the computist $A$ based on his experiences of the procedures cannot say that the universe-computer of $S$ is UC or $UC'$.  
 \end{itemize}
\end{definition}
 
  We formally defined what we mean by a computer-simulation model and indistinguishability of models. In computer-simulation model theory, CSMT, the reasoner (the computist) is also involved in the model.   To show the usefulness and importance of CSMT, we study the famous problem P vs NP.   Using CSMT, in the following  of this paper, we show that there exists no proof for   $\mathrm{P=NP}$.
\begin{itemize}
\item We introduce two notions: non-predetermined functions and persistently evolutionary Turing machines. 

\item We construct a computer simulation world, named  $E$, which $\mathrm{P\neq NP}$ in this world. The UC of $E$ is a  persistently evolutionary Turing machine. 

\item We discuss that the reality is not distinguishable from a persistently evolutionary model.

\end{itemize}

\section{Non-predetermined functions}
 \noindent The most important and fundamental notion of mathematics is function. A function is a process associating each element $x$ of a set $X$, to a single element $f(x)$ of another set $Y$. Classically, we assumed that all functions in mathematics are pre-determined.

\noindent In this section, we discuss functions that are not pre-determined and they are eventually determined through the way we start to associate $f(x)$ for every element $x\in X$.

\noindent We introduce Persistently Evolutionary Turing machines that compute non-predetermined functions.

\noindent We also discuss that if the UC of our reality persistently evolves then our reality will have alternate (alternate realities are worlds that could exist next to, in parallel of, or in place of our own, if we interacted with the world differently).

Let $f$ be a process that associates elements of a set $X$ to the elements of another set $Y$. If the process $f$ works well-defined then we know $f$ as a mathematical function.
But being well-defined does not force the process $f$ to be predetermined. 

Suppose that $x_1$ and $x_2$ are two different elements of $X$. I want to use the process $f$ to determine the value of $f$ for $x_1$ and $x_2$. It is up to me to first perform the process $f$ on $x_1$ or $x_2$. 

If $f$ is predetermined the it does not matter to perform the process on ordering $x_1x_2$ or ordering $x_2x_1$. But if $f$ is non-predetermined then different order of inputs  causes different \textit{alternate functions} which one of them is the function that we are constructing.

 Alternate functions are functions that could exist in place of our function (if we interacted with different ordering of inputs, those alternate could happen).

For example, consider the following process $g$:
\begin{itemize}
    \item \textit{$W$ is a set which is initially empty.}
    
    \item \textit{for a given natural number $n$, if there exists a pair $(n,z) \in W$ then output $g(n)=z$, else update $W=W\cup\{(n,|W|+1)\}$ and output $g(n)=|W|+1$.}
\end{itemize}

The function $g$ is a non-predetermined function over natural numbers. I  input $7,9,1,11$ and the process will associates the following: $g(7)=1$, $g(9)=2$, $g(1)=3$, and $g(11)=4$. The value of other numbers are yet non-predetermined and as soon as I perform process $g$ on each number the value is determined.

\begin{itemize}
    \item[-] The  function $g$ is not predetermined. It is determined eventually, but it is always  undetermined for some numbers. 
    \item[-] The function $g$ is well-defined, and associates to each input a single output.
    \item[-] For every natural number, the function $g$ is definable. 
    \item[-] If I inputted  $9,1,7,11$, I would have an alternate $g$ which would associate: $g(9)=1$, $g(1)=2$, $g(7)=3$, and $g(11)=4$.
\end{itemize}

\section{Persistently Evolutionary Turing machines}
Persistently Evolutionary Turing machines are an extension of the notion of Turing machines in which the structure of the machine can evolve through each computation.

A Turing machine consists of a set of states $Q$, and a table of transitions $\delta$ which both are fixed and remain unchanged forever. In Persistently Evolutionary Turing machines, we allow the set of states and the table of transitions changes through each computation. 

As a Persistently Evolutionary Turing Machine $PT$   computes on an input string $x$, the machine $PT$ can \underline{add} or \underline{remove} some of its states and transitions, and thus after the computation on the input $x$ is completed, the sets $Q$ and $\delta$ changed.
  
However these changes are persistent, that is, if we already input a string $x$ and the machine outputs $y$, then whenever we again input $x$ the machine outputs the same $y$, and the changes of states and transitions does not violate well-definedness.

One may consider that we have a BOX and we set a Turing machine in the box with some rules of \underline{adding} and \underline{removing} of states and transitions. Then, We input strings to the BOX and for each string, the BOX outputs a single string. The machine in the BOX changes itself but the behavior of the BOX is well-defined.
 
 Persistently Evolutionary Turing Machines computes non-predetermined functions. 
 
 In the following example, we introduce a persistently evolutionary nondeterministic finite automate~\cite{kn:mom0}.   

\begin{example}\label{autool}\emph{ (In the sequel of the paper, we will refer to the persistently evolutionary machine introduced in this example by $PT_{1}$). }
    
    \noindent Define $\mathrm{Evolve}:
    \mathrm{NFA}_1\times\Sigma^*\rightarrow \mathrm{NFA}_1$ as
    follows\footnote{$\mathrm{NFA}_1$ is the class of all
        nondeterministic finite automata   $M=\langle Q, \Sigma=\{0,1\}, \delta, q_0, F \rangle$,
        where for   each  state $q\in Q$, and $a\in \Sigma$, there
        exists at most one transition from $q$ with label $a$.}:
    
    \noindent Let $M\in\mathrm{NFA}_1$, $M=\langle Q,
    q_0,\Sigma=\{0,1\},\delta:Q\times\Sigma\rightarrow Q, F\subseteq
    Q\rangle\footnote{F is the set of accepting states}$,  and $x\in \Sigma^*$. Suppose $x=a_0a_1\cdots a_k$
    where $a_i\in \Sigma$. Applying the automata $M$ on $x$, one of
    the three following cases may happen:
    \begin{itemize}
        \item[case1.] The automata $M$   reads all $a_0,a_1\cdots ,a_k$
        successfully and stops in an accepting state. In this case,  the structure of the automata does not change   and let $\mathrm{Evolve}(M,x)=M$.
        
        \item[case2.] The automata $M$   reads all $a_0,a_1\cdots ,a_k$
        successfully and stops in a state $p$ which is not an accepting
        state.
        \begin{itemize} 
            \item  If the automata $M$ can  transit from the state $p$ to an  accepting
        state by reading \emph{one } alphabet, then  let $\mathrm{Evolve}(M,x)=M$.
         \item If it cannot transit (from $p$ to an accepting state)  then let
        $\mathrm{Evolve}(M,x)$ to be a new automata $M'=\langle Q,
        q'_0,\Sigma=\{0,1\},\delta':Q'\times\Sigma\rightarrow Q',
        F'\subseteq Q'\rangle$, where $Q'=Q$, $\delta'=\delta$,
        $F'=F\cup\{p\}$.
    \end{itemize}
        \item[case3.] The automata $M$ cannot read all $a_0,a_1\cdots ,a_k$
        successfully,and after reading a part of $x$, say $a_0a_1\cdots
        a_i$, $0\leq i\leq k$, it crashes in  a state $q$ that
        $\delta(q,a_{i+1})$ is not defined. In this case, we let $\mathrm{Evolve}(M,x)$
        be a new automata $M'=\langle Q,
        q'_0,\Sigma=\{0,1\},\delta':Q'\times\Sigma\rightarrow Q',
        F'\subseteq Q'\rangle$, where $Q'= Q\cup \{s_{i+1},s_{i+2},\cdots,
        s_k\}$ (all $s_{i+1},s_{i+2},\cdots, s_k$ are new states
        that does not belong to $Q$), $\delta'=\delta\cup
        \{(q,a_{i+1},s_{i+1}), (s_{i+1},a_{i+2},s_{i+2}),\cdots,
        (s_{k-1},a_k,s_{k})\}$, and $F'=F\cup\{ s_k\}$.
    \end{itemize}
 
\end{example}
The machine $PT_{1}$ persistently evolve, that is, if it (rejected) accepted a string $x$ already, then it would (reject) accept the string $x$ for any future trials as well. The language $L(M)$ is not predetermined and it eventually is determined.

For example, assume that initially $M$ is $Q=\{q_0\}$, $F=\emptyset$, $\delta=\emptyset$. Now I input the string $101$ and according to case~3, the machine $M$ evolves and  new states $q_1,q_2,q_3$ and transitions $(q_0,1,q_1),(q_1,0,q_2),(q_2,1,q_3)$ are added and also $F=F\cup \{q_3\}$. Now if I input the string $10$ then according to case~2, $M$ rejects it. However, If at first I inputted $10$ to the machine then it would accept it.

\subsection{Time complexity of Evolutionary Turing machines}
The time-complexity~\cite{kn:arora0} of Persistent Evolutionary Turing Machines is defined similar to the time-complexity of Turing machines except that for each (\underline{adding})  \underline{removing} of states or transitions, we count one extra clock.

\begin{proposition}\label{timece} The time complexity of the machine $PT_1$ in example~\ref{autool} is linear.
\end{proposition}\begin{proof}
It is straightforward. 
\end{proof}

\subsection{Executable By Computer Devices}
Persistently Evolutionary Turing machines (similar to Turing machines) are executable by computer devices if as soon as the device needed extra memory resources, we are ready to add memory cards to the motherboard of the device.

 Two simulations $V$  and $E$ introduced in example~\ref{tce}, and in definition~\ref{Ee} are computer simulations. The $SBOX$ and $TBOX$ of the $V$ are Turing computable and the $SBOX$ and $TBOX$ od $E$ are Persistently Evolutionary Turing computable.
\section{Alternate Reality} \label{ARe}

Alternate realities are worlds that could exist next to, in parallel of, or in place of our own, if we interacted with the world differently.

If the UC of our reality persistently evolves then we have alternate realities. If this morning, I first had breakfast and then washed my hands, I would be in another alternate reality, but since I first washed my hands and then had breakfast, I am in the alternate world which we call it   "reality". 

Although there are lots of alternate realities, only one of them is actualized. UC evolves as human beings (the one who lives in this simulation)  interact with it. The future is not predetermined and via our interactions, we always are moving to new alternate realities.

\section{ $\mathrm{P= NP}$ Contradicts with Non-predeterminism}
 In following, I construct a computer simulation and show that in this simulation $\mathrm{P}$ is not equal to $\mathrm{NP}$.
  
   As, we discussed the first section, assuming that the reality is a computer simulation then all we (human beings) do (thinking, walking, reasoning, talking, jumping, and etc) are procedures performed by a computer, which we called it "Universe-Computer" (UC). When I think, it is UC that computes, when I prove a theorem, it is UC that computes, when I breathe, it is UC, and etc. Physics, Biology, Logic and etc are all computational process which UC performs. So, we may name human beings (the one who lives in) to be the \underline{computist} as his/her thinking, breathing, walking and etc are all computations done by UC.
 
 What if UC persistently evolves? Then multiverse and alternate realities are possible (see section~\ref{ARe}), and it is not predetermined that we are moving to which alternate realities, and future alternate realities are eventually determined by our interaction with UC. 

 In the next section, we construct a  computer simulation that its UC persistently evolves and it is not predetermined that what alternate realities happen as we move to the future. We show that $\mathrm{P=NP}$ contradicts with non-predeterminism and thus in this simulation $\mathrm{P}$ is not equal to $\mathrm{NP}$.
 
\section{A Computer Simulation which $\mathrm{P=NP}$ does not hold in} 
 
 According to the reasoning method $\mathrm{CSMT} $, If we want to show that $\mathrm{P=NP}$ is not provable, we need to construct a simulation  such that 
 \begin{itemize}
     \item[1.] we (as the one who lives in) cannot distinguish the simulation from the reality.
     \item[2.] $\mathrm{P=NP}$ does not hold in the simulation
 \end{itemize}

 We first, in example~\ref{tce}, introduce a computer-simulation model named $V$. Then in definition~\ref{Ee}, we slightly change the UC of $V$ and construct a computer-simulation model named $E$ which its UC persistently evolves. Then in theorem~\ref{subex}, we prove that $\mathrm{ P \neq NP}$ in $E$.

\begin{example}\label{tce} We introduce a computer-simulation model $V$ as follows:

Let 
\begin{itemize}
    \item[] $Q_T=\{h\}\cup\{q_i\mid   i\in \mathbb{N}\cup\{0\}\}$,
    \item[] $\Sigma,\Gamma$ be two finite set  with $\Sigma\subseteq \Gamma$
    and
    \item[]  $\Gamma$ has a symbol $\triangle \in \Gamma-\Sigma$.
\end{itemize}

The \emph{UC} of the simulation $V$,  $U_v=(TBOX_v,SBOX_v, INST_v, CONF_v)$ is
defined as follows:
\begin{itemize}
\item[1)] $INST_v=\{[(q,a)\rightarrow(p,b,D)]\mid p,q\in Q_T,
a,b\in \Gamma, D\in\{R,L\}\}$,

$(INST_v)_0=\{[(q,a)\rightarrow(p,b,D)]\in INST_s\mid q=q_0 \}$,
and

%$[(q,a)\rightarrow(p,b,D)]\simeq %[(q',a')\rightarrow(p',b',D')]$
%iff $q=q'$ and $a=a'$.

\item[2)] $CONF_v=\{(q,x\underline{a}z)\mid q\in Q_T, x,z\in
\Gamma^*, a\in \Gamma\}$,   for each $x\in \Sigma^*$,
$C_{0,x}=(q_0,\underline{\triangle} x)$, and for each
$C=(q,x\underline{a}z)\in CONF_s$, $y_C=xaz$.

\end{itemize}
Note that the programming language of $U_v$ is exactly the
standard syntax of configurations and transition functions of
Turing machines.
\begin{itemize}

\item[3)] Let $C=(q,xb_1\underline{a}b_2y)$ be an arbitrary
configuration then

\begin{itemize} \item $TBOX_v(C,[(q,a)\rightarrow (p,c,R)])$ is defined to be $C'=(p,xb_1c\underline{b_2}y)$,

\item $TBOX_v(C,[(q,a)\rightarrow (p,c,L)])$ is defined to be
$C'=(p,x\underline{b_1}cb_2y)$, and

\item for other cases $TBOX_v$ is defined to be $\bot$.

\end{itemize}
\item[4-] Let $C\in CONF_v$ be arbitrary
\begin{itemize}
\item if $C=(h,\underline{\triangle}x)$ then $SBOX_v(C)$ is
defined to be $YES$,

\item if $C=(h,x\underline{\triangle})$ then $SBOX_v(C)$ is
defined to be $YES$, and

\item otherwise $SBOX_v(C)$ is defined to be $NO$.
\end{itemize}

\item[5-] For each $M\in \Xi_v$, and $C=(q,x\underline{a}y)\in
CONF_v$,   if there exists $[(q,a)\rightarrow(p,b,D)]\in M$ for
some $p\in Q_T, b\in \Gamma$, and $D\in \{R,L\}$, then
$\upsilon(M,C)$ is defined to be $[(q,a)\rightarrow(p,b,D)]$ else
it is defined to be  $\bot$.
\end{itemize}
 
\end{example}

\begin{theorem}\label{lang} Accepting the Church-Turing thesis,  the computist in the simulation $V$, can compute all the procedures that we (human beings) can compute in reality ~\footnote{If we accept Church-Turing thesis, all the procedures in the reality are Turing computable.}.

\begin{itemize}
\item[1-]For each  procedure $M\in \Xi_v$, there exists a
Turing machine $T$ such that for every $x\in
\Sigma^*$, $x\in L(M)$ and $time_M(x)=n$ iff $x\in L(T)$ and
$time_T(x)=n$.

\item[2-] For each Turing machine $T$, there exists a
 procedure $M\in \Xi_s$,  such that for
every $x\in \Sigma^*$, $x\in L(T)$ and $time_T(x)=n$ iff $x\in
L(M)$ and $time_M(x)=n$.

\end{itemize}
\end{theorem}\begin{proof} It is straightforward for us not for the computist who lives inside the simulation $V$.\end{proof}

We as the one who creates the simulation $V$ can prove the above theorem since we know the inner structure of $TBOX_v$ and $SBOX_v$. The computist who
lives inside the simulation cannot be aware of the above theorem as he
does not have access to the inner structure of the \emph{UV}. However, he/she is always free to propose a  thesis
about his/her world (simulation)   similar to what we did, and named it Church-Turing thesis~\footnote{
Note that we (the human beings) do not have access to the
inner structure of the reality, we call the Church Turing statement to be a thesis.}.

\begin{definition}
A \emph{UC}, $U=(TBOX,SBOX, INST, CONF)$ is called to be
static, whenever the inner structure of two boxes $TBOX$ and
$SBOX$ does not change due to interaction with the computist. It
is called persistently evolutionary whenever the  inner
structure of at least one boxes changes but persistently, i.e., in
the way that the boxes work well-defined.
\end{definition} The \emph{UC}  defined in the example~\ref{tce} is static.

\begin{definition}\label{Ee} We introduce a computer simulation $E$ which the \emph{UC} of $E$, $$U_e=(TBOX_e,SBOX_e,  INST_e,CONF_e)$$ is defined as follows. \begin{itemize}

\item[] Two sets $INST_e$ and $CONF_e$ are defined to be the same
$INST_s$ and $CONF_s$ in example~\ref{tce} respectively, and consequently the set of
procedures of the  $U_e$, i.e., $\Xi_e$ is the
same $\Xi_s$.

\item[] The transition box $TBOX_e$ is also defined similar to the
transition box $TBOX_s$ in example~\ref{tce}.

\item[] The successful box $SBOX_e$ is defined as follows: let
$C\in CONF_e$ be arbitrary
\begin{itemize}
\item if $C=(h,\underline{\triangle}x)$ then $SBOX_s(C)=YES$,

\item if $C=(h,x\underline{\triangle})$ then the $SBOX_e$ works
exactly similar to the   the persistently evolutionary
 machine $PT_{1}$ introduced in example~\ref{autool}. On input $x$,  if $PT_1$ outputs $1$, the successful box
outputs $YES$, and

\item otherwise $SBOX_e(C)=NO$.
\end{itemize}
\end{itemize}
\end{definition}

\begin{proposition}
    Computer-simulations $E$ and $V$ satisfies conditions $c1-c12$ of definition~\ref{CE}.
\end{proposition}\begin{proof}
We only need to discuss $c12$ for the simulation $E$. By proposition~\ref{timece}, the UC of $E$ works in linear time. 
\end{proof}

  Note that the $SBOX_e$  of the $U_e$ is
a persistently evolutionary machine. For the computist
  the set of  procedures (algorithms) in the simulation
$E$ is the same set of  procedures in the simulation
$V$, i.e $\Xi_v=\Xi_e$. However for some procedures, say $M$, the language $L(M)$ is the simulation $E$ could be different from the language $L(M)$ in $V$. For some $M\in \Xi_e$, we have $L(M)$ is a non-predetermined language. The procedure $M$ is fixed and does not change through time, but since the structure of UC ($SBOX_e$)  changes through time.

The computist in the simulation $E$ thinks that everything is static since the structure of the procedure does not change and the computist also does not have access to the structure of the UC of the simulation.

\begin{theorem}\label{prf}   Every recursively  enumerable language can be computed in the environment
$E_e$. That is, for every recursively enumerable language $L$,
there exists $M\in \Xi_e$ such that $L=L(M)$.
\end{theorem}\begin{proof}
 It is straightforward. For each Turing machine $T$, one may
 construct a  procedure  $M\in\Xi_e$ such that $L(T)=L(M)$.
\end{proof}

\begin{proposition}\label{pee}   The complexity class $\mathrm{P}$ is a subset of
$\mathrm{P}_{E}$.
\end{proposition}\begin{proof}
It is straightforward.\end{proof}

The converse of   Theorems~\ref{prf}, and~\ref{pee} do not hold
true. In the simulation $E$, the computist can compute some languages in
polynomial time which are not predetermined (see
the proof of theorem~\ref{subex}).

  The item c1 of
 definition~\ref{CE} says that the computist is free in the way of his/her interactions with the UC. In the simulation $E$, the UC  is a  persistently evolutionary machine and based on different orderings of the interactions of the computist, we would have different alternate futures.

\begin{definition}
We say a function $f:\mathbb{N}\rightarrow \mathbb{N}$ is
sub-exponential, whenever there exists $t\in \mathbb{N}$ such that
for all $n>t$, $f(n)<2^n$.
\end{definition}

\begin{theorem}\label{subex} % \textbf{(GV)}. 
In the simulation $E$, 
\begin{center}
    $\mathrm{NP_E \neq P_E}$.
\end{center}
That is, there exists a  procedure $M\in \Xi_e$ such that

\begin{itemize}
\item the language $L(M)$ that the computist  
computes  through $M$ is not predetermined,

\item the language $L(M)$  belongs to the class
$P_{E}$, 

\item there exists no  procedure $M'\in \Xi_e$, such that
 $L(M')$  is equal to $L'= \{x\in \Sigma^*\mid\exists y
(|y|=|x|\wedge y\in L(M))\}$, and for some $k\in \mathbb{N}$, for
all $x\in L(M')$, if $|x|>k$ then
\begin{center}
$time_M(x)\leq f(|x|)$
\end{center}
where $f:\mathbb{N}\rightarrow \mathrm{N}$ is a sub-exponential
function. In other world, $L'$ is in $\mathrm{NP_E}$ but not in $\mathrm{P_E}$.

\end{itemize}
\end{theorem}\begin{proof}
 Consider the following procedure $M\in \Xi_e$
\begin{itemize} \item[] $\Sigma=\{0,1\},\Gamma=\{0,1,\triangle\},$

\item[]$M=\{[(q_0,\triangle)\rightarrow(h,\triangle,R)],
[(h,0)\rightarrow(h,0,R)],[(h,1)\rightarrow(h,1,R)]\}$.
\end{itemize} The language of the   procedure $M$, $L(M)$, is not predetermined in simulation $E$. That is, as the computist  chooses a string $x\in
\Sigma^*$ to check whether $x$ is an element of   $L(M)$,  the inner structure of the universe-computer, $U_e$ evolves. Depending on the ordering of the strings, says $x_1,x_2,...$, that the computist chooses to check whether $x_i\in L(M)$ the language $L(M)$ eventually is determined.

 It is obvious that the language  $L(M)$ belongs to $\mathrm{P}_{E}$
(due to the definition of time complexity in definition~\ref{CE}).

 Let $L'=\{x\in \Sigma^*\mid\exists y
(|y|=|x|\wedge y\in L(M))\}$. It is again obvious that $L'$
belongs to $\mathrm{NP}_{E}$.

Suppose there  exists  a  procedure $M'\in \Xi_e$ that the
computist can compute $L'$ by $M'$ in time complexity less than a
sub-exponential function $f$. Then for some $k\in \mathbb{N}$, for
all $x$ with length greater than $k$, $x$ belongs to $L'$
whenever\begin{itemize} \item[] the computist constructs  a
successful computation path $C_{0,x}C_{1,x},...,C_{n,x}$ of the
 procedure $M'$ on $x$, for some $n\leq
f(|x|)$.\end{itemize}

 Let $m_1\in \mathbb{N}$ be the maximum length of those
strings $y\in \Sigma^*$ that   \underline{until now} are accepted
by the persistently evolutionary   machine $PT_1$ (see example~\ref{autool}) which is
inside the $SBOX_e$ of $U_e$. Define $m=\max(m_1,k)$.

For every $y\in \Sigma^*$, let
$path(y):=C_{0,y}C_{1,y},...,C_{f(|y|),y}$ be the computational
path of  the procedure $M'$ on the string $y$. The $path(y)$ is  generated by the transition box of $U_e$. 
%Note that $TBOX_e$ is
%a static machine and does not evolve, thus %$path(y)$ is
%independent of the behavior of the computist.
 Let
\begin{center}
$S(y)=\{C_{j,y}\mid C_{j,y}\in path(y) \wedge \exists x\in
\Sigma^*(C_{i,y}=(h,x\underline{\triangle} ))\}$
\end{center}and
\begin{center}
$H(y)=\{x\in \Sigma^*\mid \exists C_{j,y}\in path(y)
(C_{j,y}=(h,x\underline{\triangle} ))\}$
\end{center}
We refer by $|H(y)|$ to the number of elements of $H(y)$, we have
$|H(y)|\leq f(|y|)$ if $|y|>k$. Also let $E(y)=H(y)\cap \{x\in
\Sigma^*\mid |x|=|y|\}$, and $D(y)=H(y)\cap \{x\in \Sigma^*\mid
|x|=|y|+2\}$.

\begin{center}
    
\end{center}

\noindent Let $w\in \Sigma^*$ with $|w|>m$ be arbitrary. Two cases are possible:
either $S(w)=\emptyset$ or $S(w)\neq \emptyset$.

\underline{Consider the first case}. $S(w)=\emptyset$.

The computist wants to check if the string $w$ is in $L(M')$, and  the  UC  of the simulation $E$, that is $U_e$, works on (the
 procedure $M'$ and the string $w$), to compute whether
$w$ is in $L(M')$ or not. Since the set $S(w)$ is empty, the
execution of $M'$ on $w$ does not make the $SBOX_e$ to evolve, and
it remains unchanged.
\begin{itemize}
    \item If the computist in the simulation $E$, computes that $w\in L(M')$ then it means that
    there exists a string $v\in \Sigma^*$ such that
    
    \begin{center}
         $|v|=|w|$ and
        $v\in L(M)$ $(*)$.
    \end{center}  
     
    \begin{itemize}
        \item[(i)] Since the
        computist  has the free will (see c1 of definition~\ref{CE}), he first starts to computes     procedure $M$ on all strings in $\Sigma^*$
        with length $|v|+1$ sequentially. As the length of $v$ is greater
        than $m$, all strings with length $|v|+1$ are accepted by the
        persistently evolutionary Turing machine $PT_1$ (see item-3 of
        example~\ref{autool}) which is inside  $SBOX_e$. 
        
        \item[(ii)]    Then the computist   checks that whether $v$ is $L(M)$. But because of the evolution of $SBOX_e$ happened in part (i), the  UC on the computation of $M$ on $v$ outputs $NO$, and  thus
        $v$ is not an element of    $L(M)$ in the simulation $E$ (see the item-2 of
        example~\ref{autool}). So $v\not\in L(M)$, and it contradicts with $(*)$.
    \end{itemize}

    \item If the computist  in the simulation $E$, computes  that $w\not\in L(M')$ then it means
    that for all strings $v\in \Sigma^*$, $|v|=|w|$, we have $v\not\in
    L(M)$. But it contradicts with  the free will of the computist
    again. As the length $w$ is greater than $m$, the computist  may
    choose a string $z$ with $|z|=|w|$ and by the item-3 of
    example~\ref{autool}, we have $z\in L(M)$, contradiction.
\end{itemize}

\underline{Consider the second case}. $S(w)\neq\emptyset$.

 Suppose that the computist, before computing $M'$ on $w$,   starts to compute the procedure  $M$ on all strings $v0$'s, for all  $v\in E(w)$,
and then   computes procedure
$M$ on all strings $v0$'s, for all $v\in D(w)$ respectively.

 Since
$|w|>m$, the computist   have $u0\in L(M)$ for all $u\in
E(w)\cup D(w)$, and  $SBOX_e$ of
$U_e$ evolves through computing $M$ on $u0$'s.
The   UC   evolves
in the way that   $SBOX_e$ outputs $No$ for all
configuration in 
$$\{C_{i,w}\in S(w)\mid \exists x\in E(w)\cup
D(w)(C_{i,w}=(h,x\underline{\triangle}))\}.$$

 After that, the computist starts to compute $M'$ on $w$. 
 Either the computist
finds $w\in L(M')$ or $w\not\in L(M')$.

\begin{itemize}
    \item  Suppose the first case
    happens and $w\in L(M')$. It contradicts with the free will of the
    computist. The computist computes  $M$
    on all strings $v0$, $|v|=|w|$ sequentially, and would make $\{v0\in
    \Sigma^*\mid |v|=|w|\}\subseteq L(M)$. Then the $SBOX_e$
    evolves in the way that, it will output $No$ for all
    configurations $(h,v\underline{\triangle})$, $|v|=|w|$, and   thus
    there would exist no $v\in L(M)\cap\{x\in \Sigma^*\mid |x|=|w|\}$ which implies $w\not\in L(M')$, contradiction.
    
    \item 
    Suppose the second case happens and $w\not\in L(M')$. Since
    $|H(w)|<f(|w|)< 2^{|w|}$, during the computation of $M'$ on $w$,
    only $f(|w|)$ numbers of configurations  of the form
    $(h,x\underline{\triangle})$, $x\in\{v0\mid |v|=|w|\}\cup \{v1\mid
    |v|=|w|\}$ are given as input to the $SBOX_e$. Therefore
    there exists a string $z\in \{x\in \Sigma^*\mid |x|=|w|\}$ such
    that none of its successors have been input to the persistently
    evolutionary Turing machine $PT_1$, and if the computist chooses
    $z$ and computes $M$ on it, then $z\in L(M)$ which implies $w \in L'$. Contradiction.
\end{itemize}

 We showed that $L'$ cannot be computed by any $M'$ that its time complexity is less that a sub-exponential function. Thus $L'$ does not belong to the class $\mathrm{P}$. But because of the procedure $M$, we have $L'$ belongs to $\mathrm{NP}$ and therefore in the simulation $E$,
 
 \begin{center}
     $\mathrm{P\neq NP}$.
 \end{center}
 
\end{proof}

Note that the proof of Theorem~\ref{subex} cannot be carried by the computist who lives in the environment. The above proof is done by us  (the creator of the simulation $E$). 
The   theorem simply says that if $L'$ belongs to $\mathrm{NP_E}$  then it forces the computist  to
 interacts with the $UC$ of the simulation $E$ in some certain orders, which conflicts with the item c1 of definition~\ref{CE}.
 
% \begin{thebibliography}{10}
 % \end{thebibliography}

\section{Reality  is not distinguishable from a Persistently Evolutionary Simulation Model}

   \noindent We, inhabitants of reality, can never find out whether the reality persistently evolves or not. We can never discover that whether the UC  of the reality is a Turing machine or a Persistently Evolutionary Turing machine.

In previous sections, we introduce two computer simulations $E$ and $V$ which the set of their procedures are the same, i.e., $\Xi_v=\Xi_e$. The UC of $V$ is static and the UC of $E$ persistently evolves, though the computist cannot be aware that whether the UC of $V$ ($E$) is static or persistently evolves.

Procedures in the simulation $E$ (procedures in the set $\Xi_e$) are fixed and do not change, but for some $M\in \Xi_e$, we have $L(M)$ is non-predetermined due to the evolution of the UC of $E$ (see~\ref{subex}). We as the creator of the simulation $E$, we know that the language of the procedure $M$ in the proof of~\ref{subex} is non-predetermined but the computist who lives inside $E$ cannot be sure that $L(M)$ is not predetermined and he/she may think that $L(M)$ is static.

We are also the computist of reality, and we do not have access to the UC of reality. The UC of reality is a black box for us. Actually, the UC of any computer simulation (video game) is a black box for its inhabitants.

The set of the procedures of the reality is the set of all Turing machines, i.e. $\Xi_{reality}$ is the set of all Turing machines. 

The structure of Turing machines is fixed and does not change through computations (similar to the procedures of simulations $V$ and $E$), but we (human beings, inhabitants of the reality) cannot be sure that the language of Turing machines are predetermined.

 At any moment of time, we only interact with a finite number of times with the UC, and based on a finite number of interactions, it is not possible to discover that if the UC is static or persistently evolves.

\begin{definition}  \emph{(BLACK BOX)}. Let $X$ and $Y$ be two
sets,

\begin{itemize}
\item  an \emph{input-output} black box $\mathbf{B}$, for an
observer, is a box that \begin{itemize} \item  The observer
  does not see the inner instruction of the box, and

\item the observer   chooses elements in $X$, and input them to
the box, and receives elements in $Y$ as output.
\end{itemize}\item   We say an input-output black box $\mathbf{B}$ behaves
\underline{well-defined}  whenever if the
observer   inputs $x_0$ to the black box, and the black box
outputs $y_0$ at a stage of time, then whenever in future if the
observer  inputs the same $x_0$ again, the black box outputs
the same $y_0$.

\item We say a well-defined black box is static (not
order-sensitive) whenever for all $n\in\mathbb{ N}$, for every
$x_1,x_2,...,x_n\in X$, for every permutation $\sigma$ on
$\{1,2,...,n\}$, if the  observer   inputs $x_1,x_2, ...,x_n$
respectively to $B$ once, and receives
$y_1=\mathbf{B}x_1,y_2=\mathbf{B}x_2,...,y_n=\mathbf{B}x_n$, and
then  we reset the 
black box. After resetting, if the observer inputs
$x_{\sigma(1)},x_{\sigma(2)},$$...,x_{\sigma(n)}$ respectively to
$\mathbf{B}$, then the outputs of $\mathbf{B}$ for each $x_i$
would be the same already output $y_i$, before reset.

\item If an observer cannot reset a black box $\mathbf{B}$, then he/she can never discover
whether the black box is static or order-sensitive.
\end{itemize}
\end{definition}

\begin{proposition}~
\begin{itemize}
\item For every finite sets of pairs $S=\{(i_k,o_k)\mid 1\leq
k\leq n, n\in \mathbb{N}, i_k,o_k\in \Sigma^*\}$, there exists a
Turing machine $T$ such that for all $(i_k,o_k)\in S$, if we give
$i_k$ as an input to $T$, the Turing machine $T$ outputs $o_k$.

\item For every finite sets of pairs $S=\{(i_k,o_k)\mid 1\leq
k\leq n, n\in \mathbb{N}, i_k,o_k\in \Sigma^*\}$, there exists a
Persistent evolutionary Turing machine $N$ such that for all
$(i_k,o_k)\in S$, if we give $i_k$ as an input to $N$, the
Persistent evolutionary  machine $N$ outputs $o_k$.

\end{itemize}
\end{proposition}\begin{proof}
It is straightforward.
\end{proof}

\begin{corollary}
Let $B$ be an input-output black box for an observer. At each
stage of time, the observer has observed only a finite set of
input-output pairs. By the previous proposition, at each stage of
time, the observer knows both the following cases to be possible:
\begin{itemize}
\item[1-] There exists a Turing machine inside the black box $B$.

\item[2-] There exists a Persistent evolutionary Turing machine
inside the black Box $B$.
\end{itemize}
\end{corollary}

Therefore, We  (which do not have access to the UC of the reality) can never discover that whether the UC of the reality is static or order-sensitive and persistently evolves.

The UC  of the simulation $E$ ($SOBX_e$ and $TBOX_e$) works in linear time. As we do not have access to the inner side of the UC of reality, we cannot say whether the "successful box" of the reality evolves or not. The universe-computer, UC, of the reality is a linear-time oracle for us (as the computist), and a linear-time oracle does not affect complexity classes.  
\begin{center}
We can never distinguish the reality from $E$.
\end{center}
  
 \section{Conclusion}
 \begin{itemize}
     \item[1-]  We introduced a new method of reasoning named $\mathrm{CSMT}$.
     \item[2-] We constructed a computer simulation $E$ which its universe-computer, UC, persistently evolves. 
     \item[3-] We proved, in the computer simulation $E$, $\mathrm{P}$ is not equal to $\mathrm{NP}$. The simulation $E$ is a counter-model for $\mathrm{P=NP}$.
     \item[4-] We (who lives in the reality) does not have access to the UC of the reality, it is a black box for us, and we can never discover that whether the UC of the reality is static or persistently evolves. 
     \item[5-] We cannot prove $\mathrm{P=NP}$, since if we could prove $\mathrm{P=NP}$, then we could discover that the UC of the reality does not persistently evolve, it contradicts with item 4.
 \end{itemize}

%\begin{thebibliography}{10}
%\end{thebibliography}

\end{document}